\documentclass[12pt, a4paper]{amsart}
\usepackage{amsmath,amssymb,latexsym}

\newcommand{\re}{\text{\rm Re\,}}

\newcommand{\bd}{{\mathbb{D}}}

\newcommand{\bc}{{\mathbb{C}}}

\newcommand{\bt}{{\mathbb{T}}}

\newcommand{\ca}{{\mathcal{A}}}
\newcommand{\css}{{\mathcal{S}}}

\renewcommand{\a}{\alpha}

\renewcommand{\l}{\lambda}
\newcommand{\e}{\varepsilon}
\newcommand{\s}{\sigma}
\renewcommand{\sp}{\sigma_p(J)}

\newcommand{\p}{\varphi}

\renewcommand{\th}{\theta}
\renewcommand{\d}{\delta}
\newcommand{\dd}{\Delta}
\renewcommand{\o}{\omega}
\newcommand{\oo}{\Omega}
\newcommand{\g}{\gamma}
\renewcommand{\gg}{\Gamma}
\newcommand{\ep}{\varepsilon}
\newcommand{\z}{\zeta}

\newcommand{\nt}{\noindent}

\newcommand{\bsl}{\backslash}

\newcommand{\pt}{\partial}

\newcommand{\lp}{\left(}
\newcommand{\rp}{\right)}

\newcommand{\wt}{\widetilde}
\newcommand{\dsp}{\displaystyle}

\DeclareMathOperator{\dist}{\rm dist}

\allowdisplaybreaks
\numberwithin{equation}{section}

\newtheorem{theorem}{Theorem}[section]

\newtheorem{lemma}[theorem]{Lemma}

\theoremstyle{definition}

\begin{document}

\title[A Blaschke-type condition and its application]
{A Blaschke-type condition and its application to complex Jacobi matrices}
\author[A. Borichev, L. Golinskii, S. Kupin]{A. Borichev, L. Golinskii, S.
Kupin}

\address{Universit\'e Aix-Marseille, 39, rue Joliot-Curie, 13453
Marseille \\Cedex 13, France}
\email{borichev@cmi.univ-mrs.fr, kupin@cmi.univ-mrs.fr}

\address{Mathematics Division, Institute for Low Temperature Physics and
Engineering, 47 Lenin ave., Kharkov 61103, Ukraine}
\email{leonid.golinskii@gmail.com}

\date{December 1, 2007}


\begin{abstract}
We obtain a Blaschke-type necessary condition on zeros of analytic functions on
the unit disk with different types of exponential growth
at the boundary.  These conditions are used to prove  Lieb-Thirring-type
inequalities for the eigenvalues of complex Jacobi matrices.
\end{abstract}

\maketitle

\vspace{-0.5cm}
\section*{Introduction}
\label{s0}

In the first part of the paper, we obtain some information on the distribution
of the zeros of analytic functions from special growth classes. We use this
information to get interesting counterparts of famous Lieb-Thirring inequalities
\cite{li1,li2} for complex Jacobi matrices.

The traditional approach to Lieb-Thirring bounds for complex Jacobi matrices
consists in deducing them from the bounds for corresponding self-adjoint
objects, see Frank-Laptev-Lieb-Seiringer \cite{la} and Golinskii-Kupin
\cite{gk1}. However, this method is quite limited. At best, it allows us to get
an information on a part of the point spectrum $\s_p(J)$ of a Jacobi matrix $J$,
situated in a very special diamond-shaped region, see \cite[Theorem 1.5]{gk1}.
The information on the whole $\s_p(J)$ is missing.

The main idea of our paper is to use functional-theoretic tools
in the problem described above. Let  $J=J(\{a_k\},\{b_k\},\{c_k\})$
be a complex Jacobi matrix \eqref{e63} such that $J-J_0$ lies in
the Schatten-von Neumann class $\css_p$, $p\ge1$, and
$J_0=J(\{1\}, \{0\}, \{1\})$ (see Section \ref{s3} for terminology).
For an integer $p$, the regularized perturbation determinant
$u_p(\lambda)=\det_p(J-\lambda)(J_0-\lambda)^{-1}$ is well-defined, and is
an analytic function on $\hat\bc\bsl \s(J_0)=\hat\bc\bsl [-2,2]$.
Its zero set coincides with $\s_p(J)$ up to multiplicities.
We can obtain some information on the distribution of
the zeros of the function $u_p$ from its growth estimates in a
neighborhood of the boundary of  $\hat\bc\bsl [-2,2]$. As usual,
the domain $\hat\bc\bsl [-2,2]$ is mapped conformally to the unit disk $\bd$, so
we are mainly interested in properties of corresponding analytic functions on
$\bd$.

Let $\ca(\bd)$ be the set of analytic functions on the unit disk
$\bd$, $f\in\ca(\bd)$, $f~\not=0$, and $Z_f=\{z_j\}$ denote the set
of zeros of $f$.

\begin{theorem}\label{t1} Given a {\it finite} set $E=\{\z_j\}_{j=1,\dots,N}$, $E\subset\bt$,
let  $f\in\ca(\bd)$,  $|f(0)|~=1$,  and

\begin{equation}
|f(z)|\le \exp\Bigl(\frac{D}{\dist(z,E)^q}\Bigr),\label{1}
\end{equation}
with $q\ge 0$.
Then for any $\ep>0$,
\begin{equation}\label{e22}
\sum_{z\in Z_f}  (1-|z|)\dist(z,E)^{(q-1+\varepsilon)_+}\le
C(\ep,q,E)\, D.
\end{equation}
\end{theorem}
Here, $x_+=\max\{x,0\}$.

Clearly, \eqref{e22} is a Blaschke-type condition. The classical
Blascke condition is valid for functions from Hardy spaces
$H^p(\bd), 0<p\le\infty$, or, more generally, from the Nevanlinna
class $\mathcal N$, see Garnett \cite[Ch. 2]{ga}, and follows from
the Poisson--Jensen formula:
\begin{equation}\label{e100} \sum_{z\in
Z_f} (1-|z|)\le \sup_r \int_{\bt}\log|f(r\z)|\,dm -\log|f(0)|,
\end{equation}
where $dm$ is the normalized Lebesgue measure on $\bt$. With small
modifications one can write its analogs for the Bergman and the
Korenblum spaces, see Hedenmalm-Korenblum-Zhu \cite{he1}. The
specifics of our particular problem lead us to consider weights with
a finite number of ``exponential singularities" at the boundary in
addition to the radial growth of the weight.

Note that for $q<1$ the function $f$ \eqref{1} is in $\mathcal N$,
and \eqref{e22} follows directly from \eqref{e100}. So in the proof
of \eqref{e22} we can assume $q\ge 1$.

\begin{theorem}\label{t2} Let  $f\in\ca(\bd),\  |f(0)|=1$,  and
\begin{equation}\label{e101}
|f(z)|\le \exp\Bigl(\frac{D_1}{(1-|z|)^p\dist(z,E)^q}\Bigr),
\end{equation}
where $p,q\ge 0$.
Then for any $\ep>0$,
\begin{equation}\label{e21}
\sum_{z\in Z_f} (1-|z|)^{p+1+\varepsilon}\dist(z,E)^{(q-1+\varepsilon)_+}\le
C(\ep,p,q, E)\, D_1.
\end{equation}
\end{theorem}

One can vary the degrees of the singularities $\z_j$ in \eqref{1}
and \eqref{e101}. More precisely,

\begin{theorem}\label{t31} Let
$$
h_1(z)=\lp\prod^N_{j=1} |z-\z_j|^{q_j}\rp^{-1},\qquad
h_2(z)=\lp(1-|z|)^p\prod^N_{j=1} |z-\z_j|^{q_j}\rp^{-1},
$$
with $\{q_j\}_{j=1,\dots, N},\ q_j\ge 0$. Furthermore, let $f_j
\in\ca(\bd)$, $|f_j(0)|=1$, $j=1,2$, satisfy
$$
|f_j(z)|\le\exp(K_j h_j(z)). 
$$
Then
\begin{eqnarray}\label{e102}
\sum_{z\in Z_{f_1}} (1-|z|)\prod^N_{j=1}
|z-w_j|^{(q_j-1+\ep)_+}&\le& C(\ep, \{q_j\},E) K_1, \\ \sum_{z\in
Z_{f_2}} (1-|z|)^{p+1+\ep}\prod^N_{j=1} |z-w_j|^{(q_j-1+\ep)_+}&\le&
C(\ep, p, \{q_j\},E) K_2.  \label{e103}
\end{eqnarray}
\end{theorem}
For the sake of simplicity, we prove relations \eqref{e22},
\eqref{e21}. The proofs of \eqref{e102}, \eqref{e103} are similar.

The starting idea of the work is close in
spirit to an interesting paper by Demuth-Katriel \cite{dek}. To obtain
counterparts of Lieb-Thirring bounds, the authors look at the
difference of two semigroups generated by two continuous
Schr\"odinger operators and they apply the classical Poisson-Jensen
formula. They work with nuclear and Hilbert-Schmidt perturbations
and the potential has to be from the Kato class. Our methods seem to
be more straightforward. The computations are simpler and they are
valid for $\css_p$-perturbations, $p\ge1$. In particular, we do not
require  the {\it self-adjointness} of the perturbed operator.

As usual, $\bd_r=\{z : |z|<r\}, \ \bd=\bd_1, \bt_r=\{z: |z|=r\},
\bt=\bt_1,$ and $B(z_0,\d)=\{z: |z-z_0|<\d\}$. $C$ is a constant
changing from one relation to another one.

We proceed as follows. In Section 2 we prove Theorem \ref{t1} and
derive Theorem \ref{t2} from it. In Section 3 we discuss
applications to complex Jacobi matrices.

 \section{Proofs of Theorems \ref{t1} and \ref{t2}}\label{s1}

Given a circular arc $\gg=[e^{it_1}, e^{it_2}]$, denote by
$\o(z;t_1,t_2)$ its harmonic measure with respect to the unit disk
$\bd$. An explicit formula is available (see Garnett \cite[Ch. 1, Exercise 3]{ga})
$$\o(z;t_1,t_2)=\frac1{\pi}\lp\a(z)-\frac{t_1-t_2}2\rp, $$
where $\gg$ is seen from a point $z$ under the angle $\a$. We use
the notation $\o_\g$ for the symmetric arc
\begin{equation*}\label{e201}
\gg=\{\z\in\bt:|1-\z|\le\g\}=[e^{-it(\g)},
e^{it(\g)}], \quad \sin\frac{t(\g)}2=\frac{\g}2.
\end{equation*}
Here $\g=\g(E)<1/500N$ is a small parameter which will depend on $E$. By
the Mean Value Theorem
\begin{equation}\label{e202}
\frac{\g}{\pi}\le\o_\g(0)=\frac{t(\g)}{\pi}\le\frac{\g}2\,.
\end{equation}
Next,
\begin{equation}\label{e203}
\o_\g(z)\ge \frac12-\frac{t(\g)}{\pi}\ge \frac14 \qquad
\end{equation}
for $|1-z|\le\g$.
An outer function
\begin{equation}\label{e204}
g_\g(z)=e^{\o_\g(z)+i\wt{\o_\g(z)}}, \quad |g_\g(\z)|=\left\{%
\begin{array}{ll}
    e, & \hbox{$|1-\z|\le\g$,} \\
    1, & \hbox{$|1-\z|>\g$,} \\
\end{array}%
\right. 
\end{equation}
where $\z\in\bt$, will play a key role in what follows. Clearly, $\o_\g=\log|g_\g|$, and
\begin{equation}\label{e205}
1\le |g_\g(z)|\le e, 
\end{equation}
for $z\in\bd$.

We need a bound for the Blaschke product
$$ b_\l(z):=\frac{z-\l}{1-\bar\l z} $$
in the case when the parameter $\l$ and variable $z$ are ``well-separated''.

\begin{lemma}\label{l1}
Let $M=200N$. Then for $|1-z|=\g$, $|1-\l|\ge M\g$ one has
\begin{equation}\label{e206}
\log\frac1{|b_\l(z)|}\le \frac1{4N\g}\,\log\frac1{|\l|}\,.
\end{equation}
\end{lemma}

\medskip
\nt{\it Proof}. Obviously, we asume $\l\not=0$. Consider the domain
$$ U:=\bd\backslash B(\l,r), \quad r=\frac{1-|\l|}4\,, $$
and two harmonic functions on $U$
$$ V_\l(z):=\log\frac1{|b_\l(z)|}, \quad
W_\l(z):=\frac{(1+\e)^2-|z|^2}{|(1+\e)w-z|^2}\,, $$ where parameters
$\e=\e(\l)>0$ and $w=w(\l)\in\bt$ are chosen later on. Clearly,
$V_\l(\z)=0<W_\l(\z)$ for $\z\in\bt$, and we want to bound these
functions on $\pt B(\l,r)$.

This is easy for $V_\l$:
$$ \frac1{|b_\l(z)|}=\frac4{1-|\l|}\,|1-\bar\l z|, $$
and for $z=\l+re^{it}$ we have $$ 1-\bar\l
z=1-|\l|^2-\frac{1-|\l|}4\,\bar\l e^{it}$$ so $|1-\bar\l
z|<3(1-|\l|)$ and
\begin{equation*}\label{e207}
V_\l(z)<\log 30<5, 
\end{equation*}
where $z\in\pt B(\l,r)$.

The problem is more delicate for the lower bound on $W_\l$. Let
 $\l=|\l|e^{i\th}$ and  $w=e^{i\p}$.
 Then, for $z\in\pt B(\l,r)$
\begin{equation}\label{e209}
(1+\e)^2-|z|^2>1-|z|^2=1-|\l|^2-r^2-2r|\l|\cos(\th-t)>\frac13(1-|\l|).
\end{equation}
Next, we want to have a bound $O((1-|\l|)^2)$ for the denominator of
$W_\l$
\begin{equation}\label{e210}
|(1+\e)w-z|^2=|w-z|^2+\e^2+2\e\re(1-\bar wz).
\end{equation}
For the first term,
\begin{align*}
|w-z|^2 &=1-2\re\bar wz+|\l|^2+r^2+2r|\l|\cos(\th-t)=S_1+S_2, \\
S_1 &=(1-|\l|)^2+r^2-2r(1-|\l|)\cos(\th-t), \\
S_2 &=2|\l|+2r\cos(\th-t)-2\re\bar wz.
\end{align*}
For $S_1$ one already has $S_1<2(1-|\l|)^2$. To get
$$ S_2=2|\l|\lp 1-\cos(\th-\p)\rp+2r\lp
\cos(\th-t)-\cos(t-\p)\rp=O\lp (1-|\l|)^2\rp,
$$
we choose $w$ in the following way. If $|\th|\ge t(M\g/2)$, we put $\p=\th$, so
$w=\l/|\l|$ and $S_2=0$. If $|\th|<t(M\g/2)$, we put
$$ \p=\left\{%
\begin{array}{rr}
    t(M\g/2), & \hbox{$0\le\th<t(M\g/2)$,} \\
    -t(M\g/2), & \hbox{$-t(M\g/2)<\th<0$.} \\
\end{array}%
\right.
$$
Then, by \eqref{e202}, $|\p-\th|<t(M\g/2)\le\pi M\g/4$. On the other
hand, by the hypothesis of lemma
\begin{align*}
M\g &\le |1-\l|=\left|1-|\l|+|\l|(1-e^{i\th})\right| \\
&\le 1-|\l|+|\th|\le 1-|\l|+\frac{\pi M\g}4,
\end{align*}
so $1-|\l|\ge M\g(1-\pi/4)$ and
$|\p-\th|\le\frac{\pi}{4-\pi}(1-|\l|)$. Hence
\begin{align*}
S_2 &\le 4\sin^2\frac{\th-\p}2+4r\left|\sin\frac{\th-\p}2 \sin\frac{\th+\p-2t}2\right| \\
&\le
\lp\frac{\pi}{4-\pi}\rp^2(1-|\l|)^2+\frac{\pi}{5(4-\pi)}(1-|\l|)^2<26(1-|\l|)^2
\end{align*}
and $|w-z|^2<28(1-|\l|)^2$.

Note also that in both cases above we have $|\p|\ge t(M\g/2)$, that
is
\begin{equation}\label{e211}
|1-w|\ge \frac{M\g}2.
\end{equation}

To fix the second term in \eqref{e210} we take $0<\e<(1-|\l|)^2$, so
$$
\e^2+2\e\re(1-\bar wz)\le 5\e<5(1-|\l|)^2, $$ and eventually
\begin{equation*}\label{e212}
|(1+\e)w-z|^2<31(1-|\l|)^2.
\end{equation*}

It follows now from the above inequality and \eqref{e209} that $W_\l$ admits
the lower bound
$$ W_\l(z)>
\frac{1/3(1-|\l|)}{31(1-|\l|)^2}>\frac1{100}\frac1{1-|\l|}, 
$$
where $z\in\pt B(\l,r)$. So,
$$ \log\frac1{|\l|}\,W_\l(z)>\frac1{100} 
$$
for these $z$.
By the Maximum Principle $V_\l<500\, W_\l$ on
$U$, and by letting $\e\to 0$ we obtain
\begin{equation*}\label{e213}
\log\frac1{|b_\l(z)|}<500\log\frac1{|\l|}\,\frac{1-|z|^2}{|w-z|^2},
\end{equation*}
where $z\in U$.

Note that by the assumption of the lemma $|1-\l|=l\ge M\g$, so
$1-|\l|\le |1-\l|=l$, and if $z\in\pt B(\l,r)$, then
\begin{equation*}\label{e214}
|1-\l-re^{it}|\ge l-\frac{1-|\l|}4\ge \frac34\,M\g>2\g,
\end{equation*}
which means that the arc $\{|1-z|=\g,\ |z|<1\}$ lies inside $U$. For
such $z$ by \eqref{e211} $|w-z|\ge|1-w|-|1-z|\ge
\bigl(M/2-1\bigr)\g$, so
$$ \frac{1-|z|^2}{|w-z|^2}<\frac2{\bigl(M/2-1\bigr)^2 \g}\,. $$
The proof is complete.  \hfill $\Box$

\medskip
An invariant form of \eqref{e206} is
\begin{equation*}\label{e2141}
\log\frac1{|b_\l(z)|}\le \frac1{4N\g}\,\log\frac1{|\l|}\,,
\end{equation*}
for any $\z\in\bt$ with the properties  $|\z-z|=\g,\  |\z-\l|\ge M\g$.

\smallskip
We decompose the unit disk into a union of disjoint sets
$$ \bd=\oo\bigcup\lp\bigcup_{n,k}\oo_{n,k}\rp, \quad
\oo_{n,k}=\{z\in\bd: 2^{-n-1}<|z-\z_k|\le 2^{-n}\},
$$
where $k=1,2,\ldots,N$, $n=L,L+1,\ldots$. Here $L=L(E)$ is a large
parameter, so that
\begin{equation}\label{e215}
{\rm dist}\,(\oo_{L,k},\z_s)\ge 2^{-L}, 
\end{equation}
where $s\not=k,\ k=1,\ldots,N$.
The latter obviously yields the same inequality for ${\rm
dist}\,(\oo_{n,k},\z_s)$ for all $n\ge L$.

Let us fix a pair $(n,k)$, and define numbers
$$ \g_s=\left\{%
\begin{array}{ll}
    2^{-L-1}/M, & \hbox{$s\not=k$,} \\
    2^{-n-1}/M, & \hbox{$s=k$,} \\
\end{array}%
\right. 
$$
where $s=1,\ldots,N$, and $M$ is from Lemma \ref{l1}. Now
\eqref{e215} reads
\begin{equation}\label{e216}
{\rm dist}\,(\oo_{n,k},\z_s)\ge M\g_s. 
\end{equation}

Define three sets of arcs
$$ \gg_s:=\pt B(\z_s,\g_s)\bigcap\bd, \qquad
\tilde\gg_s:=\{\z\in\bt:|\z-\z_s|\le\g_s\}, $$ the arcs
$\tilde\gg_s$ and $\tilde\gg_{s+1}$ are separated by
$\hat\gg_s\subset\bt$, $s=1,\ldots,N$. Set  $\dd_{n,k}\subset\bd$ in a way that
$$\pt\dd_{n,k}=
\lp\bigcup_{s=1}^N \gg_s\rp\bigcup\lp\bigcup_{s=1}^N \hat\gg_s\rp.
$$

Let $\{z_j\}_{j=1}^m$ be a finite number of zeros of $f$ (counting multiplicity) in
$\oo_{n,k}$, and
$$ b_j:=b_{z_j}(z)=\frac{z-z_j}{1-\bar z_j z}\,, 
$$
for $j=1,\ldots,m$.
Consider the functions
\begin{align*}
g_s(z) &=g_{\g_s}^{\rho_s}(z\bar\z_s), \qquad \ \ \ \rho_s
=\frac{4D}{\g_s^q}\,, \\
g_{j,s}(z) &=g_{\g_s}^{\rho_{j,s}}(z\bar\z_s), \qquad \rho_{j,s}
=\frac1{N\g_s}\,\log\frac1{|z_j|}\,,
\end{align*}
with $s=1,\ldots,N$, $j=1,\ldots,m$. Recall that $g_s$ are defined in \eqref{e204}.
For $z\in\gg_s$ we have
$$
\log|b_j(z)g_{j,s}(z)|=\log|b_j(z)|+\rho_{j,s}\log|g_{\g_s}(z\bar\z_s)|.
$$
By \eqref{e216} we can apply Lemma \ref{l1}, which along with
\eqref{e203} gives
$$
\log|b_j(z)g_{j,s}(z)|\ge\frac1{4N\g_s}\log|z_j|+\frac1{4N\g_s}\log\frac1{|z_j|}=0,
$$
so
\begin{equation}\label{e217}
|b_j(z)g_{j,s}(z)|\ge 1, 
\end{equation}
for $z\in\gg_s,\  j=1,\ldots,m$.

Note also, that, for $z\in\bd$ (see \eqref{e205})
$$ |b_j(z)g_{j,s}(z)|\le e. 
$$

Next, by \eqref{e203}
\begin{equation}\label{e218}
\log|g_s(z)|=\rho_s\log|g_{\g_s}(z\bar\z_s)|\ge
\frac{D}{\g_s^q}=\frac{D}{|z-\z_s|^q}, 
\end{equation}
where $z\in\gg_s$.

\begin{lemma}\label{l2}
A function
\begin{equation}\label{e219}
F(z):=\frac{f(z)}{\prod_{j=1}^m b_j(z)\hat g_j(z) \, \prod_{l=1}^N
g_l(z)}\,, \quad \hat g_j(z):=\prod_{l=1}^N g_{j,l}(z)
\end{equation}
is analytic on $\bd$, and satisfies
\begin{equation}\label{e220}
\log|F(z)|\le 0, 
\end{equation}
for $z\in\bigcup_{l=1}^N\gg_l$,
\begin{equation} \label{e221}
\log|F(\z)|\le\frac{D}{{\rm dist}(\z,E)^q}, 
\end{equation}
for $\z\in\bigcup_{l=1}^N\hat\gg_l$.
Furthermore,
\begin{equation}\label{e222}
\log|F(0)|\ge
\frac12\,\sum_{j=1}^m\log\frac1{|z_j|}-2D\sum_{l=1}^N\g_l^{1-q}.
\end{equation}
\end{lemma}

\medskip
\nt{\it Proof}.  For $z\in\gg_s$, $s=1,\ldots,N$, we have
$$
\log|f(z)|\le \frac{D}{\dist(z,E)^q}=\frac{D}{|z-\z_s|^q},
$$
and $ |b_j|=|\hat g_j|=|g_s|=1$ for the rest of the boundary of $\dd_{n,k}$. So bounds \eqref{e220} and \eqref{e221} follow immediately from definition \eqref{e219}.

To prove \eqref{e222}, we write
$$ \log|F(0)|=-\sum_{j=1}^m\Bigl(\log|b_j(0)|+\log|\hat
g_j(0)|\Bigr) - \sum_{l=1}^N\log|g_l(0)|. $$
Apply the bound for the harmonic measure (see \eqref{e202})
\begin{align*}
\log|g_l(0)| &\le \frac{4D}{\g_l^q}\,\frac{\g_l}2=2D\g_l^{1-q}, \\
\log|\hat g_j(0)| &= \sum_{l=1}^N\rho_{j,l}\log|g_{\g_l}(0)|\le
\frac1N \sum_{l=1}^N
\frac1{\g_l}\log\frac1{|z_j|}\,\frac{\g_l}2=\frac12\log\frac1{|z_j|},
\end{align*}
so
$$ \log|b_j(0)|+\log|\hat g_j(0)|\le
\log|z_j|+\frac12\log\frac1{|z_j|}= -\frac12\log\frac1{|z_j|}.$$ The
proof of the lemma is complete. \hfill $\Box$

\medskip
\nt{\it Proof of Theorem \ref{t1}}. Define an outer function $F^*$ in
$\bd$ by its boundary values
$$
|F^*(\z)|=\left\{%
\begin{array}{ll}
    1, & \hbox{$\z\in\bigcup_{l=1}^N \tilde\gg_l$,} \\
    \dsp\exp\Bigl(\frac{D}{{\rm dist}(\z,E)^q}\Bigr), &
    \hbox{$\z\in\bigcup_{l=1}^N \hat\gg_l$.} \\
\end{array}%
\right.
$$
As $|F^*|\ge1$ on $\bt$, then $|F^*|\ge1$ in the whole disk, so by
Lemma \ref{l2} $|F|\le|F^*|$ on $\pt\dd_{n,k}$, and hence by the
Maximum Modulus Principle $|F|\le|F^*|$ in $\dd_{n,k}$. In
particular,
$$ \log|F(0)|\le\log|F^*(0)|. $$

The upper bound for the RHS follows from
\begin{equation*}\label{e223}
\log|F^*(0)|=
\int_{\bt}\log|F^*(\z)|\,dm=D\sum_{l=1}^N\int_{\hat\gg_l}
\dist(\z,E)^{-q}\,dm.
\end{equation*}
It is easy to estimate a typical integral in the above sum (note that
by the definition $\g_k\le\g_l$ for all $l=1,\ldots,N$)
$$
\int_{\hat\gg_l}\dist(\z,E)^{-q}\,dm\le\left\{%
\begin{array}{ll}
    \frac{2\pi^{q-1}}{q-1}\lp\frac1{\g_k}\rp^{q-1}, & \hbox{$q>1$,} \\
    \log\frac1{\g_k}, & \hbox{$q=1$.} \\
\end{array}%
\right. 
$$ Hence for $q>1$ (for $q=1$ the
argument is the same)
$$ \log|F^*(0)|\le C(q)ND\lp\frac1{\g_k}\rp^{q-1}, $$
so by \eqref{e222}
\begin{equation}\label{e224}
\sum_{j=1}^m (1-|z_j|)\le \sum_{j=1}^m
\log\frac1{|z_j|}<C(q,N)D\lp\frac1{\g_k}\rp^{q-1}.
\end{equation}

By the definition of $\g_k$ and $\oo_{n,k}$ we have
$$ 2M\g_k=2^{-n}\ge |z_j-\z_k|=\dist(z_j,E), $$
so \eqref{e224} implies
\begin{eqnarray*}
\sum_{j=1}^m(1-|z_j|)\dist(z_j,E)^{q-1+\e}&\le&
\sum_{j=1}^m(1-|z_j|)(2M\g_k)^{q-1+\e}\\
&\le& CD(2M\g_k)^{\e}=CD2^{-n\e}.
\end{eqnarray*}
Summation over $n\ge L$, and then over $k=1,\ldots,N$ gives
\begin{equation*}\label{e225}
\sum_{z_j\in\oo}(1-|z_j|)\dist(z_j,E)^{q-1+\e}\le CD, \qquad
\end{equation*}
where $\oo:=\bigcup_{n,k}\oo_{n,k}$.

The same reasoning with $\g_s=1/2^LM$, $s=1,\ldots,N$ applies to the
domain $\oo_0:=\{z\in\bd:|z-\z_k|>2^{-L}, \ k=1,\ldots,N\}$. The
proof is complete.  \hfill $\Box$

\bigskip\nt
{\it Proof of Theorem \ref{t2}.}\ Denote $\tau_n:=1-2^{-n}$,
$n=0,1\ldots$ and put $f_n(z):=f(\tau_n z)$. Then
 $$ |f_n(z)|\le \exp\Bigl(\frac{D_1}{(1-|\tau_n z|)^p\dist(\tau_n
 z,E)^q}\Bigr). $$
An elementary inequality
\begin{equation}\label{e226}
\frac{1-|z|}{1-\tau|z|}\le\frac{|z-\z|}{|\tau
z-\z|}\le\frac{1+|z|}{1+\tau|z|}<2,
\end{equation}
which holds for $z\in\bd$, $\z\in\bt$ and $0\le\tau<1$, gives
$$ \dist(\tau_n z,E)>\frac12 \dist(z,E) $$
so
$$ |f_n(z)|\le \exp\Bigl(\frac{D_2}{\dist(z,E)^q}\Bigr),
$$
where $D_2=2^{np+q}D_1$.
If $Z_f=\{z_j\}$, $Z_{f_n}=\{z_{j,n}\}$
then
$$ z_{j,n}=\frac{z_j}{\tau_n}:\ |z_j|<\tau_n. $$

By Theorem \ref{t1}
\begin{equation}\label{e227}
\sum_{z_j:|z_j|<\tau_n} \Bigl(1-\frac{|z_j|}{\tau_n}\Bigr)
\dist\Bigl(\frac{z_j}{\tau_n},E\Bigr)^r\le C(\e,q,E)2^{np+q}D_1,
\end{equation}
$r=(q-1+\e)_+$. We readily continue as
\begin{eqnarray*}
{\rm LHS \ of \ \eqref{e227}}&\ge&\sum_{\tau_{n-2}\le|z_j|<\tau_{n-1}} \Bigl(1-\frac{|z_j|}{\tau_n}\Bigr)
\dist\Bigl(\frac{z_j}{\tau_n},E\Bigr)^r\\
&\ge& \frac1{4^{r+1}}\sum_{\tau_{n-2}\le|z_j|<\tau_{n-1}}(1-|z_j|)
\dist(z_j,E)^r,
\end{eqnarray*}
and, consequently,
$$ 2^{-n(p+\e)} \sum_{\tau_{n-2}\le|z_j|<\tau_{n-1}}(1-|z_j|)
\dist(z_j,E)^r\le C2^{-n\e}D_1.
$$
Since $1-|z_j|\le 1-\tau_{n-2}$ then
$$ 2^{-n(p+\e)}\ge \frac1{4^{p+\e}}(1-|z_j|)^{p+\e}, $$
and finally
$$ \sum_{\tau_{n-2}\le|z_j|<\tau_{n-1}}(1-|z_j|)^{p+\e+1}
\dist(z_j,E)^r\le C2^{-n\e}D_1. $$ It remains only to sum up over
$n$ from $2$ to $\infty$. \hfill $\Box$

\section{Applications to complex Jacobi matrices}\label{s3}
We are interested in complex-valued Jacobi matrices of the form
\begin{equation}\label{e63}
J=J(\{a_k\},\{b_k\},\{c_k\})=
\begin{bmatrix}
b_1&c_1&0&\ldots\\
a_1&b_2&c_2&\ldots \\
0&a_2&b_3&\ldots \\
\vdots&\vdots&\vdots&\ddots
\end{bmatrix}
\end{equation}
where $a_k, b_k, c_k\in\bc$. We assume $J$ to be a compact
perturbation of  the free Jacobi matrix $J_0=J(\{1\},\{0\},\{1\})$, or,
equivalently,
$\lim_{k\to+\infty} a_k=\lim_{k\to+\infty} c_k=1$, $ \lim_{k\to+\infty} b_k=0$.
It is well-known that in this situation
$\s_{ess}(J)=[-2,2]$. The point spectrum of $J$ is
denoted by $\s_p(J)$; the eigenvalues $\lambda\in\s_p(J)$  have finite algebraic (and
geometric) multiplicity, and the set of their limit points lies on the interval
[-2,2] (see, e.g., \cite[Lemma I.5.2]{gk}).

The structure of $\s_p(J)$ and, especially, its behavior near $[-2,2]$, is an
important part of the spectral analysis of complex Jacobi matrices.  We quote a
theorem from Golinskii-Kupin \cite{gk1}  to give a flavor of the known results.
\begin{theorem}\label{t3} For $p\ge 1$,
\begin{eqnarray*}
\qquad \sum_{\lambda\in\s_p(J)} (\re\lambda-2)^p_+&+&\sum_{\lambda\in\sp} (\re\lambda+2)^p_-\\
 &\le& c_p\Bigl(\sum_{k=1}^\infty |\re b_k|^{p+1/2}+4\Bigl|\frac{a_k+\bar
c_k}2-1\Bigr|^{p+1/2}\Bigr),
\nonumber\\
\qquad\sum_{\lambda\in\s_p(J)} (\re\lambda-2)^p_+&+&\sum_{\lambda\in\sp} (\re\lambda+2)^p_-\\
 &\le& 3^{p-1}\Bigl(\sum_{k=1}^\infty |\re b_k|^p+4\Bigl|\frac{a_k+\bar
c_k}2-1\Bigr|^p\Bigr),
\nonumber
\end{eqnarray*}
where $x_+=\max\{x,0\}$, $x_-=-\min\{x,0\}$, and
\begin{equation}\label{e10}
c_p=\frac{3^{p-1/2}}2 \frac{\gg(p+1)}{\gg(p+3/2)}\frac{\gg(2)}{\gg(3/2)}.
\end{equation}
\end{theorem}

To formulate the results of this section, we use the Schatten-von
Neumann classes of compact operators $\css_p$, $p\ge1$, and
regularized determinants $\det_p (I+A)$, $A\in\css_p$. An extensive
information on the subject is in Gohberg-Krein \cite[Ch. 3, 4]{gk},
Simon \cite[Ch. 2, 9]{si1}. The norms in $\css_p$, $1\le p<\infty$,
$\css_\infty$, are denoted by $\|\cdot\|_p$, $\|\cdot\|$,
respectively.

Let $\lambda\in\hat\bc\bsl[-2,2]$. We map this domain onto $\bd$ in the standard way
$\lambda=z+1/z,\ z=\frac 12(\lambda-\sqrt{\lambda^2-4}), z\in\bd$.  Here is a list of
elementary properties of the above-mentioned concepts and their connections to
the spectral characteristics of operator $J$. Suppose $J-J_0\in\css_p, 1\le
p<\infty$.
\begin{itemize}
\item[--] $(J-J_0)(J_0-\lambda)^{-1}\in \css_p.$
\item[--] for an integer $p\ge1$,
$$
u_p(\lambda)=\det{}_p(J-\lambda)(J_0-\lambda)^{-1}=\det{}_p(I +(J-J_0)(J_0-\lambda)^{-1}),
$$
and so the regularized perturbation determinant $u_p$ is well-defined.
\item[--] the function $u_p$ is analytic on $\hat\bc\bsl[-2,2]$. Furthermore,
$Z_{u_p}=\s_p(J)$ taking into account the multiplicities, i.e., the order of a
zero $\lambda_0$ of the function $u_p$ is equal to the total multiplicity of
$\lambda_0\in\s_p(J)$.
\item[--] We also have
\begin{eqnarray*}
|\det{}_p(J-\lambda)(J_0-\lambda)^{-1}|&\le&\exp\Bigl(\frac 1p
\|(J-J_0)(J_0-\lambda)^{-1}\|^p_p\Bigr)\\
&\le&\exp\Bigl(\frac 1p \|(J-J_0)\|^p_p\, \|(J_0-\lambda)^{-1}\|^p\Bigr)\\
&=&\exp\Bigl(\frac 1p \|(J-J_0)\|^p_p\, \mathrm{dist}\, (\lambda,
[-2,2])^{-p}\Bigr).
\end{eqnarray*}
The first inequality is in \cite[Ch.9]{si1}. Then we recall that
$\css_p$ is an ideal with respect to the multiplication and use the
expression for the resolvent $\|(J_0-\lambda)^{-1}\|$ of a {\it
self-adjoint} operator $J_0$.
\end{itemize}

\begin{lemma}\label{l3} Let $\lambda=z+1/z$, $z\in\bd\bsl \bd_{\d}$, where $0<\d<1$.
Then
\begin{eqnarray*}
\mathrm{dist}\, (\lambda, [-2,2])&\asymp& (1-|z|)|1-z^2|,\\
|1\pm z|^2&\asymp& |\lambda\pm 2|,\quad
1-|z|\asymp \frac{\mathrm{dist}\, (\lambda, [-2,2])}{|\lambda^2-4|^{1/2}}.
\end{eqnarray*}
\end{lemma}

For the proof of the first relation see \cite[p.9, Corollary 1.4]{pom}.

Let $f_p(z)=u_p(\lambda(z))$, $p>1$. Then
\begin{equation*}
|f_p(z)|\le\exp\Bigl(\frac{\|J-J_0\|^p_p}{p(1-|z|)^p|1-z^2|^p}\Bigr).
\end{equation*}
For $p=1$, it is proved in \cite{eg2} that
\begin{equation*}
|f_1(z)|\le
\frac{2\|J-J_0\|_1}{|1-z^2|}\exp\Bigl(\frac{2\|J-J_0\|_1}{|1-z^2|}\Bigr).
\end{equation*}

\begin{theorem}\label{t4}
For $p=1$ and every $\ep>0$ we have
\begin{equation}\label{e8}
\sum_{\lambda\in\s_p(J)} \frac{\mathrm{dist}\, (\lambda,
[-2,2])}{|\lambda^2-4|^{(1-\ep)/2}}\le C(\ep, \|J-J_0\|,
p)\|J-J_0\|_1.
\end{equation}

For integer $p\ge 2$ and every $\ep>0$ we have
\begin{equation}\label{e81}
\sum_{\lambda\in\s_p(J)} \frac{\mathrm{dist}\, (\lambda,
[-2,2])^{p+1+\ep}} {|\lambda^2-4|}\le C(\ep, \|J-J_0\|, p)
\|J-J_0\|^p_p.
\end{equation}
\end{theorem}

\begin{proof} Indeed, $\sigma_p(J)$ is in $B(0, 2+\|J-J_0\|)$, so $Z_{f_p}$ lies in
$\bd\bsl\bd_{\d}$ with $0<\d<1$ depending on $\|J-J_0\|$.

For instance, to get the second relation we apply Theorem \ref{t1} to $f_p$:
\begin{multline*}
\sum_{\lambda\in\s_p(J)}(1-|z(\lambda)|)^2 \mathrm{dist}\, (\lambda, [-2,2])^{p-1+\ep}\\
=\sum_{z\in Z_{f_p}} (1-|z|)^{p+1+\ep}|1-z^2|^{p-1+\ep}\le
C(\ep,\|J-J_0\|, p) \|J-J_0\|^p_p.
\end{multline*}
It remains to use Lemma \ref{l3}. The first relation follows in a similar way
from Theorem \ref{t2}.
\end{proof}

It is worth mentioning that the only ingredient we need to make the proof of
Theorem \ref{t4} work, is the bound on the resolvent of the unperturbed operator.
Neither its self-adjoint property, nor three-diagonal form are required. For
instance, we can prove the same result for an operator, similar to $J_0$ and its perturbations.

Although we do not claim the results being optimal, the bounds give
a lot of interesting information. As compared to classical
Lieb-Thirring inequalities for complex Jacobi matrices, the bounds
\eqref{e8} and \eqref{e81} involve the whole point spectrum
$\s_p(J)$, and not of its relatively simple parts (see \cite{gk1}).
In particular, we see that $\s_p(J)$ behaves differently along the
interval $(-2,2)$ and in the neighborhoods of its endpoints $\pm2$.

An analogous theorem holds, of course, for multidimensional Jacobi
matrices, see \cite[Sect. 2]{gk1} for definitions.

\end{document}